\documentclass[aps,prl,superscriptaddress,twocolumn]{revtex4-1}

\usepackage[utf8,latin1]{inputenc}
\usepackage[T1]{fontenc}     
\usepackage[british]{babel}
\usepackage[dvipsnames]{xcolor}
\usepackage[colorlinks=true,citecolor=Green,linkcolor=Red,urlcolor=Cyan,hyperindex]{hyperref}
\usepackage{graphicx}
\usepackage{epsfig}
\usepackage{lipsum}     
\usepackage{color}
\usepackage{lmodern}
\usepackage{mathpazo}
\usepackage[babel=true]{microtype}
\usepackage{amsmath,amssymb,amsthm}
\usepackage[capitalize]{cleveref}
\usepackage{mathtools}
\usepackage{bm}
\usepackage{soul}
\usepackage{enumerate}

\newtheorem{theorem}{Theorem}
\newtheorem*{theorem*}{Theorem}

\newtheorem{definition}{Definition}

\newcommand{\ket}[1]{\vert#1\rangle}

\newcommand{\ketbra}[2]{\vert #1 \rangle \langle #2 \vert}

\newcommand{\tr}{\text{\normalfont Tr}}

\newcommand{\beq}{\begin{equation}}
\newcommand{\eeq}{\end{equation}}

\newcommand{\LL}{\text{Herm}}
\newcommand{\MM}{\mathbb{M}}

\newcommand{\cH}{\mathcal{H}}
\newcommand{\cS}{\mathcal{S}}
\newcommand{\cM}{\mathcal{M}}

\newcommand{\bM}{\textbf{M}}
\newcommand{\bN}{\textbf{N}}

\newcommand{\II}{\mathbb{I}}
\newcommand{\CC}{\mathbb{C}}

\allowdisplaybreaks

\begin{document}

\title{Distributed sampling, quantum communication witnesses, and measurement incompatibility}

\author{Leonardo Guerini}
\email{guerini.leonardo@ictp-saifr.org}
\affiliation{International Centre for Theoretical Physics - South American Institute for Fundamental Research \& Instituto de F\'isica Te\'orica - UNESP, R. Dr. Bento Teobaldo Ferraz 271, S\~ao Paulo, Brazil}

\author{Marco T\'{u}lio Quintino}
\affiliation{Department of Physics, Graduate School of Science, The University of Tokyo, Hongo 7-3-1, Bunkyo-ku, Tokyo 113-0033, Japan}

\author{Leandro Aolita}
\affiliation{Instituto de F\'isica, Universidade Federal do Rio de Janeiro, P. O. Box 68528, Rio de Janeiro, RJ 21941-972, Brazil}

\begin{abstract}
We study prepare-and-measure experiments where the sender (Alice) receives trusted quantum inputs but has an untrusted state-preparation device and the receiver (Bob) has a fully-untrusted measurement device.
A distributed-sampling task naturally arises in such scenario, where the goal is for Alice and Bob to reproduce the statistics of his measurements on her quantum inputs using a fixed communication channel.  
Their performance at such task can certify quantum communication (QC), and this is formalised by measurement-device-independent QC witnesses. 
Furthermore, we prove that QC can provide an advantage (over classical communication) for distributed sampling if and only if Bob's measurements are incompatible. 
This gives a novel operational interpretation to measurement incompatibility, and motivates a generalised notion of it related to a subset of quantum states. 
Our findings have both fundamental and applied implications.
\end{abstract}
\maketitle

The prepare-and-measure (PM) scenario is an ubiquitous framework to investigate several foundational and communicational problems.
There, one has two distant parties, Alice and Bob, and a referee, who sends them classical random inputs.
Accordingly, Alice prepares a physical system, encoding a (classical or quantum) message that she sends to Bob. Bob then makes a measurement on the system and returns his outcome to the referee for final analysis.
Depending on whether the message is classical or quantum, this framework provides a natural mindset for, \textit{e.g.}, classical and quantum dimension witnesses \cite{wehner2008, gallego2010}, quantum key distribution \cite{pawlowski2011}, classical and quantum random access codes \cite{ambainis1999, pawlowski2009}, and self-testing \cite{tavakoli2018, farkas2019}.
All these tasks have been extensively studied in the so-called device-independent (DI) paradigm, where both Alice's state-preparation and Bob's measurement stations are given by untrusted apparatuses effectively treated as black-box devices (the dimension of the communication channel is sometimes assumed, though). 
This implies that both devices admit only classical inputs and that Bob's device generates only classical outputs.

Alternatively, partially DI paradigms have also proven to yield extremely fruitful research lines. 
These consist of settings where the devices have both trusted (\textit{i.e.} well-characterised and with full quantum control) and untrusted components.
Notable instances thereof are the phenomena of Einstein-Podolsky-Rosen (EPR) steering \cite{wiseman2007}, semi-quantum instrumental causal networks \cite{nery2018}, nonlocal correlations with quantum inputs \cite{buscemi2012} (which can be interpreted as measurement-device-independent entanglement certification \cite{branciard2010}), and certification of quantum memories \cite{rosset2018}. 
These studies have revealed interesting aspects of quantum theory that could not be properly addressed in the DI regime.

Here, we study the PM scenario with quantum inputs for Alice, which we call \emph{semi-quantum prepare-and-measure} (SQPM).
More precisely, we consider a hybrid device for Alice, that admits trusted quantum-state preparations as inputs but is measurement-DI, and a fully-DI black-box device for Bob (see Fig. \ref{fig:scenario}). 
\begin{figure}[b!]
\includegraphics[width=.7\linewidth]{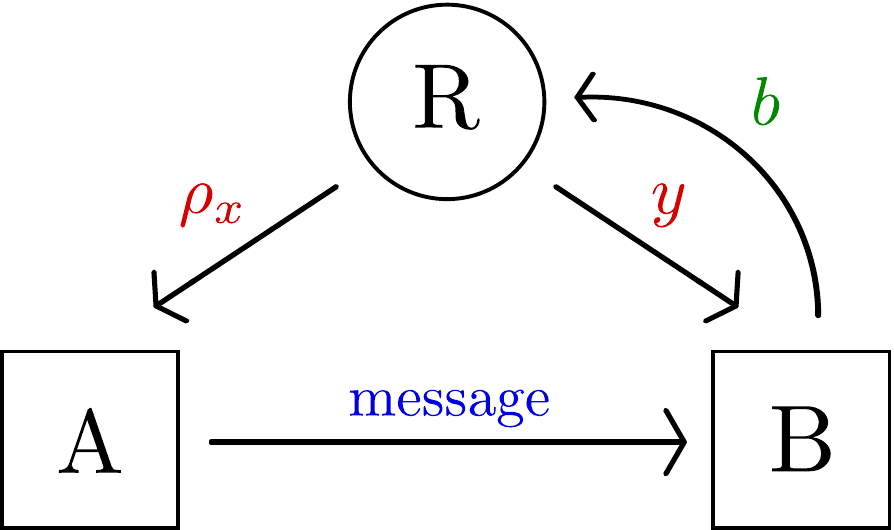}
\caption{The semi-quantum prepare-and-measure scenario: A referee sends a random state $\rho_x$ to Alice (unknown to her) and a random label $y$ to Bob. 
Alice then sends a $\rho_x$-dependent message to Bob, who makes a $y$-dependent measurement on it and returns the outcome $b$ to the referee.
}
  \label{fig:scenario}
\end{figure}
In contrast to the usual PM scenario with classical inputs, in SQPM not all well-defined statistics admit a physical realisation. 
Our first contribution is thus to characterise the set of SQPM statistics that arise from quantum experiments. 
Then, we introduce a distributed-sampling (DS) problem where the goal is for Bob to simulate the outputs (\textit{i.e.} sample from the outcome distribution) of measurements associated to his inputs
on Alice's quantum states, using as little communication as possible. 
This is an information-theoretic task that can be used to certify quantum communication from Alice to Bob in the SQPM scenario. 
We formalise this through the notion of measurement-DI \emph{quantum communication witnesses}, which can be efficiently obtained by means of semi-definite programmes (SDPs).
Furthermore, DS also turns out to be intimately connected to the fundamental problem of quantum measurement incompatibility \cite{heinosaari2016}: We prove an equivalence between the quantum-communication advantage for DS (over classical communication) and the incompatibility of the measurements implemented by Bob's black box. 
This provides a precise operational interpretation of measurement incompatibility, and naturally leads to a generalised definition of compatibility relative to the input states in DS.

\textit{Preliminaries.---} 
Let $\mathcal{H}$ be a $d$-dimensional complex Hilbert space and $\LL(\cH)$ the set of linear operators acting on $\cH$.
Let $[n]=\{1,\ldots,n\}$.
The states of a quantum system associated to $\cH$ are given by linear operators $\rho\in\LL(\cH)$ that are positive semidefinite and have unit trace, $\rho\geq0,\ \tr(\rho)=1$.
The quantum measurements with $o$ outcomes on this system are described by collections $\bM_y=\{M_{b|y}\}_{b\in[o]}\subset\LL(\cH)$ of positive semidefinite operators acting on $\cH$ that sum up to the identity, $M_{b|y}\geq0,\ \sum_b M_{b|y} = \II$.
We denote by $\mathbb{S}(\cH)$ and $\MM(\cH)$ the sets of all quantum states and measurements (with any number of outcomes) on $\cH$, respectively.

A set of $m$ quantum measurements $\cM = \{\bM_{1},\ldots,\bM_{m}\}$ with $o$ outcomes is said to be \textit{compatible}, or \textit{jointly measurable} \cite{heinosaari2016}, if there exists a so-called mother measurement $\bN=\{N_a\}_{a\in[n]}$ and response functions $f(\cdot|y, a):[o]\rightarrow[0,1]$, with $f(b|y,a)\geq0$ and $\sum_b f(b|y,a)=1$ for all $(b,y,a)$, such that
\begin{equation}\label{eq:jm}
M_{b|y} = \sum_{a=1}^{n} N_a f(b|y,a),
\end{equation}
for all $y\in[m]$ and $b\in[o]$.
This expresses the fact that one can perform $\bN$ and, depending on $y$ and the mother measurement's outcome $a$ obtained, sample $b$ from $f(\cdot|y,a)$ to determine an outcome for $\bM_{y}$.
Denoting the (convex) set of compatible measurements by \textbf{COMP}, we define the \textit{generalised robustness of incompatibility} of $\cM$ by
\begin{align}\nonumber
R_I(\cM)=\min\{\eta;\ \{&(1-\eta)M_{b|y}+\eta Q_{b|y}\}_{b,y}\in\textbf{COMP},\\ &\forall y\in[m]\ \{Q_{b|y}\}_b\in\MM(\cH)\},
\end{align}
\textit{i.e.} the minimum amount of noise (represented by an arbitrary measurement $\textbf{Q}$) needed to turn the combined measurement compatible.

A quantum channel is a completely positive trace-preserving linear map $\Lambda: \mathbb{L}(\cH)\to \mathbb{L}(\cH)$, forming a set denoted by \textbf{CPTP}.
We say that $\Lambda$ is \textit{non-steering-breaking} if its adjoint $\Lambda^\dagger$ is \textit{incompatibility-breaking}, i.e. $\{\Lambda^\dagger(M_{b|y})\}_{b,y}\in\textbf{COMP}$ for all sets of measurements $\{M_{b|y}\}_{b,y}$ \cite{heinosaari2015}.
This follows from the fact that a set of measurements is incompatible if and only if it is useful for demonstrating EPR-steering \cite{quintino2014, uola2014}.
The \textit{generalised robustness of non-steering-breaking} of $\Lambda$ is
\begin{align}\nonumber
R_{NSB}(\Lambda)=&\min\{\eta;\ \Gamma \in \textbf{CPTP}, \forall \{M_{b|y}\}_{b,y}\subset\MM(\cH),\\
&\{[(1-\eta)\Lambda+\eta\Gamma]^\dagger(M_{b|y})\}\in\textbf{COMP} \}.
\end{align}

\emph{The semi-quantum prepare-and-measure (SQPM) scenario.---} 
Consider the scenario where a sender, Alice, receives a quantum input $\rho_x$, and a receiver, Bob, is given a classical input $y$ (see Fig. \ref{fig:scenario}). 
We denote by $\cS=\{\rho_x\}_x$ the set of quantum inputs for Alice. 
Alice's and Bob's inputs are randomly chosen from $\cS$ and $[m]$, respectively, by a referee.
Alice then prepares a (potentially quantum) message by implementing some (uncharacterised) operation on $\rho_x$ and sends it to Bob, who extracts a classical output $b$ from it through some (uncharacterised) measurement that may depend on $y$.
Such experiment is described by a state-conditioned behaviour $\{(P(b|x,y),\rho_x)\}_{b,x,y}$, where $P(b|x,y)$ represents the conditional probability of $b$ given $x$ (the classical label of Alice's quantum input) and $y$. 
The state-conditioned behaviour thus encapsulates the conditional probabilities of Bob's outcomes in explicit correspondence with the states $\rho_x$ of Alice's inputs. 
From now on, we use the short-hand notation $\{P(b|\rho_x,y)\}_{b,x,y}$ for state-conditioned behaviours and refer to them simply as behaviours.

The standard prepare-and-measure scenario is recovered in the case where the states in $\cS$ can be perfectly discriminated.  
Each choice $\cS$ of trusted quantum states originates a different instance of the scenario, which is completely defined by the triple $(\cS,m,o)$, where $m$ and $o$ fix the range of values for the labels $y$ and $b$, respectively.
Thus, the standard probability constraints $P(b|\rho_x,y)\geq0,\ \sum_b P(b|\rho_x,y)=1$, for all $b\in[o],\rho_x\in\cS,$ and $y\in[m]$, define the polytope of behaviours from this scenario, whose extremal points are the $o^{m|\cS|}$ deterministic behaviours (see Fig. \ref{fig:polytope}).

\begin{figure}
  \centering
  \includegraphics[width=.6\linewidth]{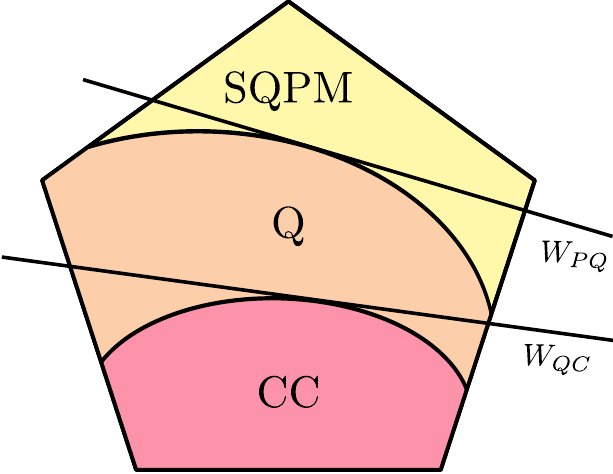}
  \caption{General behaviours from a semi-quantum prepare-and-measure scenario (SQPM) admit a polytope characterisation and are illustrated by the external pentagon. The subset of quantum behaviours (Q) can be characterised by post-quantum behaviour witnesses ($W_{PQ}$) (see App. A \cite{supp}) and the subset of behaviours that can be generated with classical communication (CC) can be characterised by quantum communication witnesses ($W_{QC}$).}
  \label{fig:polytope}
\end{figure}

The SQPM scenario is measurement-device-independent by definition.
Given $P(b|\rho_x,y)$, the conditional on the quantum state $\rho_x$ and on the classical label $y$ can be completely arbitrary; in particular we do not assume the behaviour can be obtained from quantum measurements performed on $\rho_x$.
Whenever this is the case, we say that the behaviour admits a {quantum realisation}.
\begin{definition}
A behaviour $\{P(b|\rho_x,y)\}_{b,x,y}$ \emph{admits a quantum realisation} if there exists a set of $o$-outcome measurements $\{\bM_y;\ y\in[m]\}\subset\MM(\cH)$ such that 
\begin{equation}
P(b|\rho_x,y) = \tr(\rho_xM_{b|y}),
\end{equation}
for any $b\in[o],\rho_x\in\cS,y\in[m]$. 
For simplicity, we refer to these as quantum behaviours, and denote by \textbf{Q} the set formed by them.
\end{definition}

In contrast with the prepare-in-measure scenario with classical inputs, in the SQPM not all behaviours are quantum; some statistics are incompatible with the given trusted quantum states (see App. A in the Supplemental Material \cite{supp}).
However, deciding whether there exists a quantum realisation for a given behaviour can be done efficiently by means of standard semidefinite programming (SDP), as shown in App. A.

\textit{Distributed sampling and quantum communication witnesses.---}
We now focus on quantum behaviours $\{P(b|\rho_x, y)=\tr(\rho_xM_{b|y})\}_{b,x,y}$. 
We define a {distributed sampling} task by the following rules: \begin{enumerate}
\item a referee announces a set of states $\cS$ and a set of $m$ $o$-outcome measurements $\cM$;
\item the referee sends Alice a single copy of a randomly chosen $\rho_x\in\cS$, and sends Bob a randomly chosen classical label $y\in[m]$;
\item Alice applies an arbitrary quantum operation on $\rho_x$, producing a message that is sent to Bob; 
\item conditioned on $y$ and on the message, Bob generates an output $b\in[o]$ and sends it to the referee;
\item Alice and Bob are successful if the conditional probability distributions observed by the referee after many rounds match the behaviour $\{\tr(\rho_xM_{b|y})\}_{b,x,y}$.
\end{enumerate}
Notice that $\cS$ is broadcast, but Alice does not know the particular $\rho_x$ (\textit{i.e.} the value of $x$) sent to her in each run.
The essence of this task already appears in \cite{montanaro2019}, where the inputs are not previously announced nor restricted to limited sets, two-way communication is allowed, and the main interest lies in classical communication complexity.
These differences allow us to focus on quantum properties of the involved objects.

This task is trivial if Alice and Bob have access to a perfect quantum-communication channel, namely the identity channel.
In this case, Alice can simply send $\rho_x$ to Bob, who will then hold both inputs and can implement $\bM_{y}$, reproducing the statistics accurately.
Refs. \cite{stark2016, bluhm2018} investigate the task of quantum compression, which can be interpreted as distributed sampling with perfect but lower-dimensional communication channels.

In contrast, consider now that Alice can only send classical messages to Bob. 
In this case, her most general strategy is to perform a quantum measurement $\bN=\{N_{a}\}$ on the state $\rho_x$ received by the referee and send the outcome $a$ of her measurement to Bob. 
He then outputs a classical message $b$ according to some response function, which may depend on the outcome $a$ sent by Alice and the classical input $y$ received from the referee.

\begin{definition}
A behaviour $\{P(b|\rho_x,y)\}_{b,x,y}$ \emph{admits a distributed sampling realisation with classical communication} ({CC-realisation}, for short) if there exists a quantum measurement $\bN=\{N_a\}^n_{a=1}\in\MM(\cH)$ and response functions $\{f(\cdot|y,a)\}_{y,a}$ such that 
\begin{equation}\label{eq:cc-realisation}
P(b|\rho_x,y)=\sum_{a=1}^{n} \tr (\rho_x N_a) f(b|y,a)
\end{equation}
for any $b\in[o],\rho_x\in\cS,y\in[m]$.
\end{definition}

Characterising the set $\textbf{CC}$ of behaviours that admit a CC-realisation can be done by means of an SDP. 
Moreover, such SDP can quantify how far a given behaviour $P=\{P(b|\rho_x, y)\}_{b,x,y}$ is from being CC-realisable by calculating
\begin{align}\label{gen_rob_behaviour}\nonumber
R_{NCC}(P)=\min\{\eta;\ &(1-\eta)P + \eta q\in \textbf{CC},\\ &q = \{q(b|\rho_x, y)\}_{b,x,y}\in\textbf{Q}\}.
\end{align}
\noindent We call this quantity the \textit{generalised robustness of non-CC-realisability} of the behaviour.
Also, since \textbf{CC} is convex and compact, we can describe its border by means of witnesses.
\begin{definition}
A \emph{quantum communication witness} is a pair $W_{QC}=(\{\mu_{bxy}\}_{b,x,y},\beta)$, with $\beta,\mu_{bxy}\in\mathbb{R}$, such that
\begin{equation}\label{wit:qc}
\sum_{b,x,y}\mu_{bxy}P(b|\rho_x,y) {\geq} \beta,
\end{equation}
is satisfied by all CC-realisable behaviours, but violated by some behaviour, in the scenario with trusted states $\cS=\{\rho_x\}_x$.
\end{definition}
\noindent Therefore, the violation of (\ref{wit:qc}) is a measurement-device-independent way of certifying that Alice and Bob share a quantum communication channel.
\begin{theorem}\label{thm:qc}
Let $\cS\subset\mathbb{S}(\cH)$ be a finite set of states and $o,m$ be positive integers.
Then any behaviour $P=\{P(b|\rho_x,y);\ \rho_x\in\cS,\ b\in[o],\ y\in[m]\}\notin \textbf{CC}$ violates some quantum communication witness.
Moreover, the maximal violation over all witnesses provides exactly the non-CC-realisability generalised robustness of this behaviour,
\begin{equation}
R_{NCC}(P) = \max_{W_{QC}}\sum_{b,x,y}\mu_{bxy}P(b|\rho_x,y) - \beta.
\end{equation}
\end{theorem}
\noindent In App. B \cite{supp} we provide the proof of Thm. \ref{thm:qc} and details on the SDP approach.
As an application, we study the advantage for distributed sampling of a paradigmatic noisy quantum channel over classical ones.
The results are graphically summarised in Fig. \ref{fig:depolarising}.

\begin{figure}[h!]
\centering
\includegraphics[width=1\linewidth]{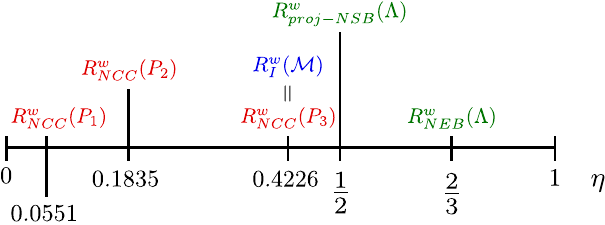}
\caption{Robustness of properties of behaviours, set of measurements, and quantum channel involved in the distributed sampling with the depolarising qubit channel $D_\eta: A\mapsto (1-\eta)A +\eta\tr(A)\II/2$ of noise $0\leq \eta\leq 1$. 
Consider the sets of states $\cS_1=\{\ketbra00, \ketbra++\}, \cS_2 = \{\ketbra00, \ketbra++,\ketbra rr\},$ and $\cS_3 = \{\ketbra00, \ketbra++,\ketbra rr, \II/2\}$, where $\ket+, \ket r$, and $\ket0$ are respectively the positive-eigenvalue eigenstates of the Pauli matrices $\sigma_x$, $\sigma_y$, and $\sigma_z$, and the measurement set $\cM\equiv \{\sigma_x, \sigma_y, \sigma_z\}$. 
The figure shows the robustness of the non-CC-realisability of the resulting behaviours $P_1, P_2$, and $P_3$, i.e. the critical noise at which the behaviour admits a distributed sampling realisation with classical communication. 
Since $D_\eta(\cM)$ is compatible for $\eta\geq1-1/\sqrt3 \approx 0.4226$, for this range $D_\eta$ is replaceable by a classical channel regardless of the input-states.
In turn, for $\eta\geq1/2$, 
the adjoint channel breaks the incompatibility of any set of projective measurements, and for $\eta\geq 2/3$ the channel is entanglement-breaking. 
More details in App. E \cite{supp}.}
  \label{fig:depolarising}
\end{figure}

\textit{Measurement incompatibility.---}
Next, we show that measurement compatibility is a "classical" property that matches precisely the classical communication case of distributed sampling.
This implies 
that any distributedly sampled quantum behaviour can be used to estimate both the degree of incompatibility of the implemented measurements and the degree of non-steering-breaking of the utilised channel.

\begin{theorem} \label{thm:jm}
A set of measurements $\cM\subset\MM(\cH)$ is compatible if and only if the behaviour $\{\tr (\rho_x M_{b|y});\ \rho_x\in\cS,\ \bM_y\in\cM\}$ admits a distributed sampling realisation with classical communication, for any set of states $\cS\subset\mathbb{S}(\cH)$ that spans $\LL(\cH)$.
Morevover, for any distributedly sampled quantum behaviour $P=\{P(b|\rho_x, y)=\tr(\widetilde{\Lambda}(\rho_x)\widetilde{M}_{b|y})\}\in\textbf{Q}$, we have 
\begin{equation}\label{bounds}
R_{NCC}(P) \leq R_{I}(\widetilde{\cM})\ \text{and}\ R_{NCC}(P) \leq R_{NSB}(\widetilde{\Lambda}),
\end{equation} 
where $\widetilde{\cM}$ and $\widetilde{\Lambda}$ are the uncharacterised measurements and communication channel, respectively, used in the sampling of $P$.
The equality holds in the first case if $\cS$ spans $\LL(\cH)$ and in the second case if, besides that, the measurements $\{\widetilde{\Lambda}^\dagger(\widetilde{M}_{b|y})\}_{b,y}$ present the greatest generalised robustness of incompatibility in its dimension.

\end{theorem}
\noindent The proof is presented in App. C \cite{supp}.

Theorem \ref{thm:jm} provides a novel operational interpretation for joint measurability in terms of a communicational task.
The first inequality in (\ref{bounds}) shows that any incompatible set of measurements generates some behaviour that can certify quantum communication via distributed sampling.
Similarly, the second inequality in (\ref{bounds}) implies that a channel $\Lambda$ is useless to certify quantum communication in this scenario if and only if $\Lambda^\dagger$ is incompatibility-breaking (\textit{i.e.} $\Lambda$ is steering-breaking).

A consequence of Thm. \ref{thm:jm} is that a certification of quantum communication via distributed sampling also detects the incompatibility of the implemented measurements.
Hence, quantum communication witnesses form a particular class of measurement incompatibility witnesses \cite{carmeli2018}.
In general, the latter are defined by a set of Hermitian operators $\{F_{by}\}_{b,y}$ acting on the same space as the measurements and a scalar $\gamma$ such that condition 
\begin{equation}\label{comp-wit}
\sum_{b,y}\tr(F_{by}M_{b|y})\geq \gamma
\end{equation}
is satisfied for any compatible set $\cM=\{M_{b|y}\}_{b,y}$, but violated by some incompatible set.
Our next result shows that, conversely, every measurement incompatibility witness also detects the exchange of quantum communication in the appropriate distributed-sampling context.

\begin{theorem}\label{thm:alljmwit}
For any measurement incompatibility witness $W_{MI}=(\{F_{by}\}, \gamma)$ there exists a set of states $\cS=\{\rho_x\}_x$ and a quantum communication witness $W_{QC}=(\{\mu_{bxy}\},\beta)$ that detects the incompatibility of the same sets of measurements as $W_{MI}$.
\end{theorem}
\noindent The proof can be found in App. C \cite{supp}.

\textit{Measurement compatibility on a restricted set of states.---}
Theorem \ref{thm:jm} unveils a direct connection between distributed sampling and measurement compatibility,
in which compatibility is equivalent to CC-realisability with an informationally complete set of states.
For more restricted sets of states, 
behaviours that can be CC-realised are directly connected to a relaxed notion of compatibility which we now define.
\begin{definition} \label{def:JM_res}
A set of $o$-outcome measurements $\{M_{b|y}\}_{b,y}$ is \emph{compatible on $\mathcal{S}=\{\rho_x\}_x$} if there exists a mother measurement $\bN=\{N_a\}_a$ and response functions $f(\cdot|y,a):[o]\rightarrow [0,1]$ such that
\begin{equation}
\tr(\rho_x M_{b|y}) = \sum_{a=1}^{n} \tr(\rho_x N_a) f(b|y,a)
\end{equation}
for all $\rho_x\in\cS$, $y\in[m]$, and $b\in[o]$.
\end{definition}
	
Def. \ref{def:JM_res} can be applied to all situations where the experimenter is guaranteed that all states in the experiment lie in a restricted set $\cS$, \textit{e.g.} a qubit experiment where all states involved lie in the $xz$ plane of the Bloch sphere. 
In such cases, the standard definition of joint measurability may represent an overkill, and the relaxed notion may be more suitable. 
See Fig. \ref{fig:depolarising} and App. E \cite{supp} for more on this topic.

\textit{Final discussion.---}
We have formalised the semi-quantum prepare-and-measure scenario and presented a distributed sampling task 
potentially interesting on its own beyond the scope of this work. 
In turn, this leads to non-classical communication certification via quantum communication witnesses. 
The underlying feature that allows this is the non-steering-breaking property of the communication channel, a notion strictly stronger than non-entanglement-breaking \cite{horodecki2003}.
Hence, our framework cannot certify steering-breaking channels that are non-entanglement-breaking (and, therefore, still non-classical).
On the other hand, both properties usually require entanglement for its certification, while our framework uses single systems only (see Ref. \cite{rosset2018} for non-entanglement-breaking channel certification in a scenario with two quantum inputs).


Measurement incompatibility is known to have operational interpretations in terms of EPR-steering \cite{quintino2014, uola2014} and state discrimination games with post-measurement information \cite{carmeli2018, uola2018, skrzypczyk2019}.  
Our findings also provide a novel communication task that captures precisely the essence of this property, complementing the operational interpretations from previous results (although being essentially different from them; see Appendix D \cite{supp} for a discussion on state discrimination).

Notice that if two parties share an entangled quantum state and classical communication they can simulate a quantum channel via the teleportation protocol.
Also, teleportation can be regarded as a protocol where Alice has a quantum input and Bob can do tomography on his final state.
Since in our scenario Bob's measurements are untrusted, our results point towards a realisation of quantum teleportation as device-independent as possible, where only Alice's input state is trusted.

Further open problems for future research are the characterisation of the inner structure of the SQPM polytope, and a quantitative study of the amount of classical/quantum communication required for approximate or probabilistic distributed sampling.
For the classical case, it would be interesting to see how the results of Ref. \cite{montanaro2019} relate to measurement incompatibility. 
Finally, the quantum-communication witnesses developed here are experimentally relevant and implementable with current technology.

\emph{Acknowledgements ---} The authors thank Ranieri V. Nery and Gl\'aucia Murta for fruitful discussions.
LA thanks the ICTP-SAIFR and the IFT-UNESP at São Paulo -- where this work was started and partially developed -- for the hospitality.
LG is supported by S\~ao Paulo Research Foundation (FAPESP) under grants 2016/01343-7 and 2018/04208-9.
MTQ acknowledges the Q-LEAP project of the MEXT, Japan. 
LA is financially supported by the Brazilian agencies CNPq (PQ grant No. 311416/2015-2 and INCT-IQ), FAPERJ (JCN E-26/202.701/2018), CAPES (PROCAD2013), FAPESP, and the Serrapilheira Institute (grant number Serra-1709-17173).


\section{Appendix}

All code is available at \url{www.github.com/guerinileonardo/DS_QCwit_MI}.

\section{A. Quantum realisations}

Let $\cS\in\mathbb{S}(\cH)$ be a set of quantum states and $m,o$ be positive integers.
Not all semi-quantum prepare-and-measure behaviours $\{P(b|\rho_x,y);\rho_x\in\cS, b\in[o], y\in[m]\}$ admit a quantum realisation, \textit{i.e.} can be written as $P(b|\rho_x,y)=\tr(\rho_xM_{b|y})$, for some measurements $\{M_{b|y}\}_{b,y}\subset\MM(\cH)$.
Indeed, consider $\cH=\mathbb{C}^2, \cS=\{\ketbra00, \ketbra11, \II/2\}, o=2, m=1$ and the behaviour specified by
\begin{subequations}\label{pqbehaviour}
\begin{align}
P(1|\ketbra00,1) &= P(1|\ketbra11, 1) = 0,\\
P(1|\frac{\II}{2},1) &= 1,
\end{align}
\end{subequations}
together with the normalisation constraints.
If this behaviour is quantum samplable, then there exists a quantum measurement $\bM_1=\{M_{1|1},M_{2|1}\}$ satisfying
\begin{subequations}
\begin{align}
0 =& P(1|\ketbra00,1) + P(1|\ketbra11,1) \\
=& \tr(\ketbra00 M_{1|1}) + \tr(\ketbra11 M_{1|1}) \\
=& 2\tr(\frac{\II}{2}M_{1|1}) \\
=& 2P(1|\frac{\II}{2},1) = 2,
\end{align}
\end{subequations}
an obvious contradiction.

More generally, we can decide whether a behaviour $\{P(b|\rho_x,y);\rho_x\in\cS,b\in[o],y\in[m]\}$ admits a quantum realisation with the semidefinite programme
\begin{align}\nonumber
\text{given}& \quad \{P(b|\rho_x,y)\}_{b,x,y}\\ \label{sdp:qr}
\min_{\{q(\cdot|\rho_x,y)\},\{\bM_y\}}& \quad \eta \\ \nonumber
\text{s.t.}& \quad (1-\eta)P(b|\rho_x,y) + q(b|\rho_x,y) = \\ \nonumber
&\quad\quad \tr(\rho_x M_{b|y}),\ \forall \rho_x,b,y \\ \nonumber
&\quad q(b|\rho_x,y)\geq 0,\ \forall \rho_x,b,y\\ \nonumber
&\quad \sum_b q(b|\rho_x,y) = \eta, \ \forall \rho_x,y  \\ \nonumber
&\quad M_{b|y}\geq 0,\ \forall b,y\\ \nonumber
&\quad \sum_b M_{b|y} = \II, \ \forall y.
\end{align}
The probabilities $\{q(\cdot|\rho_x,y)/\eta\}_{x,y}$ represent an arbitrary noise, which is mixed to the given behaviour until it accepts a quantum description $\{\tr(\rho_xM_{b|y})\}_{b,x,y}$.
Hence, the input-behaviour admits a quantum realisation if and only if the optimal value obtained is $\eta^*\leq0$.

For a fixed triple $(\cS, o, m)$, it follows from the convexity of the set of quantum measurements that the set of quantum behaviours is convex, besides compact. 
Hence, due to the Separating Hyperplane Theorem \cite{boyd2004} we have that such a set can be characterised by post-quantum behaviour witnesses.

\begin{definition}
A \emph{post-quantum behaviour witness} is a pair $W_{PQ}=(\{\lambda_{bxy}\},\alpha)$ formed by real coefficients $\lambda_{bxy}$ and a bound $\alpha$ such that 
\begin{equation}\label{wit:postq}
\sum_{b,x,y}\lambda_{bxy}P(b|\rho_x,y) {\geq} \alpha,
\end{equation}
is satisfied for all quantum behaviours $\{P(b|\rho_x,y)\}$, but violated by some behaviour of the scenario.
\end{definition}
\noindent Hence, even if a given behaviour $\{P(b|x,y)\}_{b,x,y}$ may arise from some quantum experiment, a violation of (\ref{wit:postq}) ensures that such experiment does not involve states $\{\rho_x\}_{x}$, calculated from the dual formulation of SDP (\ref{sdp:qr}).

We will now show that post-quantum behaviour witnesses can be associated to the generalised robustness \cite{eisert2007, cavalcanti2016} of the property of the inputs under study (in our case, it is the 'post-quantumness' of the behaviour), since we optimise over all possible noises.
Let us now fix an arbitrary noise, that is, consider $\{q(\cdot|\rho_x,y)/\eta\}_y$ not as variables but as inputs selected previously.
Such variation of SDP (\ref{sdp:qr}) admits the Lagrangian
\begin{subequations}
\begin{align}
 \mathcal{L}(\eta,\{M_{b|y}\}) &= \eta\left(1+\sum_{b,x,y}\lambda_{bxy}\left[P(b|\rho_x,y)-q(b|\rho_x,y)\right]\right) \\ &- \sum_{b,y}\tr\left(M_{b|y}\left[U_{by}+A_y -\sum_{x}\lambda_{bxy}\rho_x\right]\right)\\ &- \sum_{b,x,y}\lambda_{bxy}q(b|\rho_x,y) + \sum_y \tr(A_y),
 \end{align}
 \end{subequations}
for arbitrary scalars $\lambda_{bxy}$ and operators $U_{by}\geq0, A_y$.
{Each choice of these parameters yields an upper bound for the optimal value $t^*$ of SDP (\ref{sdp:qr}).}
Minimising over them, considering feasibility constraints and eliminating the slack variables $U_{by}$, we obtain the dual formulation\footnote{By fixing the noise we simplify SDP (\ref{sdp:qr}); leaving the noise unspecified would lead to the same dual with an extra constraint, concerning the non-negativity of variables $q(\cdot|\rho_x,y)$.} of SDP (\ref{sdp:qr}),
\begin{subequations}\label{sdp:qrwit}
\begin{align}\nonumber
\text{given}& \quad \{P(b|\rho_x,y)\}_{b,x,y}, \{q(b|\rho_x,y)\}_{b,x,y}\\ 
\max_{\{\lambda_{bxy}\},\{A_y\}}& \quad \sum_y\tr(A_y) - \sum_{bxy}\lambda_{bxy}P(b|\rho_x,y)  \\ \label{conditionqr}
\text{s.t.}& \quad \sum_x\lambda_{bxy}\rho_x \geq A_y,\ \forall b,y\\ \label{conditionqr2}
& \quad \sum_{bxy}\lambda_{bxy}\left(q(b|\rho_x,y) - P(b|\rho_x,y)\right) = 1.
\end{align}
\end{subequations}

\setcounter{theorem}{3}
Consider the following result.
\begin{theorem}\label{thm:qr}
Let $\cS\subset\mathbb{S}(\cH)$ be a finite set of states and $o,m$ be positive integers.
Then every behaviour $\{P(b|\rho_x,y);\ \rho_x\in\cS,\ b\in[o], y\in[m]\}$ that is not quantum violates some post-quantum behaviour witness.
\end{theorem}
\noindent The Separating Hyperplane Theorem suffices to prove Thm. \ref{thm:qr}.
Nonetheless, we will now present a second demonstration, more constructive, that is based on SDPs (\ref{sdp:qr}) and (\ref{sdp:qrwit}).
Duality theory says that the optimal solution for the dual problem is always an upper bound for the optimal solution of its primal.
Sometimes, however, we may have strong duality between them, meaning that the problems are such that both optimal values coincide.
A sufficient requirement for ensuring strong duality is called Slater's condition, which is satisfied whenever there is a feasible point satisfying all equality constraints and strictly satisfying the inequalities ones, for either one of the problems.
Our second proof of Thm. \ref{thm:qr} is based on proving that two SDPs above satisfy Slater's condition and are therefore strongly dual.

\begin{proof}
The behaviour $\{P(b|\rho_x,y)\}$ admits a quantum realisation if and only if SDP (\ref{sdp:qr}) yields an optimal solution $\eta^*\geq1$.
This SDP is strictly feasible: for $\eta=1$ we obtain a solution by placing $q(b|\rho_x,y)=\tr(\rho_xM_{b|y})$, for any arbitrary measurements $\{M_{b|y}\}_{b,y}$.
In particular, choosing $M_{b|y} = \II/o$ for all $b,y$ we have that $M_{b|y}$ are strictly positive operators, while $q(b|\rho_x,y)$ are strictly positive scalars.
Hence, by Slater's condition \cite{boyd2004} we have that (\ref{sdp:qr}) and (\ref{sdp:qrwit}) present strong duality.
This implies that $\{P(b|\rho_x,y)\}$ admits a quantum realisation if and only if the dual SDP (\ref{sdp:qrwit}) yields an optimal solution less or equal to zero, or equivalently
\begin{equation}
\sum_{bxy}\lambda_{bxy}^*P(b|\rho_x,y) \geq \sum_y\tr(A_y^*),
\end{equation}
where $\lambda_{bxy}^*,A_Y^*$ are provided by the optimal solution of (\ref{sdp:qrwit}).
Taking $\alpha=\sum_y\tr(A_y^*)$ concludes the proof.
\end{proof}

 
As an example, inputting the behaviour defined in Eqs. (\ref{pqbehaviour}) in SDP (\ref{sdp:qrwit}) we obtain the post-quantum realisation witness
\begin{equation}
-\frac{1}{2}\left[P(2|\rho_1,1)+P(2|\rho_2,1)\right]-P(1|\rho_3,1) \geq -1.
\end{equation}
The left-hand side of the above witness equals -2 when evaluated on behaviour (\ref{pqbehaviour}), attesting its "post-quantumnes".

\section{B. Distributed sampling with classical communication}\label{app:sdp}

Recall that a behaviour $\{P(b|\rho_x,y)\}_{b,x,y}$ admits a distributed sampling realisation with classical communication (is CC-realisable) if there exists a quantum measurement $\bN\in\MM(\cH)$ and response functions $\{f(\cdot|y,a)\}_{y,a}$ such that 
\begin{equation}
P(b|\rho_x,y)=\sum_{a=1}^{n} \tr (\rho_x N_a) f(b|y,a)
\end{equation}
for any $b\in[o],\rho_x\in\cS,y\in[m]$.
Notice that every such behaviour is quantum realisable, associated to the measurements $M_{b|y} = \sum_aN_af(b|y,a)$.
Hence, admitting a quantum realisation is a necessary condition for a behaviour to be CC-realisable.

The subset of behaviours that can be distributedly sampled with classical communication form the convex and compact set \textbf{CC}.
Consequently, Thm. 1 can also be seen as an application of the Separating Hyperplane Theorem \cite{boyd2004}.
In what follows, we provide an SDP approach to the problem and an alternative proof of Thm. 1.

Notice that we can always take Bob's response function to be deterministic by mapping its local randomness to the measurement $\bN$.
In other words, if such a strategy is possible, then it can be done with $\bN$ having at most $n=(o)^m$ outcomes, given by $\textbf{a}=a_1\ldots a_{m}$, with $a_i\in[o]$.
These outcomes already encode Bob's answer, who simply outputs the $y$-th symbol of the classical message $\textbf{a}$, which corresponds to applying the map $p(b|y,\textbf{a}) = \delta_{a_y,b}$.

Hence we can decide whether $\{P(b|\rho_x,y)\}_{b,x,y}\in\textbf{CC}$ via the generalised robustness SDP
\begin{align}\nonumber
\text{given}& \quad \{P(b|\rho_x,y)\}_{b,x,y}\\ \label{sdp:cc}
\min_{\bN, \{\widetilde M_{b|y}\}}& \quad \eta \\ \nonumber
\text{s.t.}& \quad (1-\eta)P(b|\rho_x,y) + \tr(\rho_x\widetilde M_{b|y}) = \\ \nonumber
&\quad\quad \sum_{\textbf{a}} \tr(\rho_x N_\textbf{a})\delta_{a_y,b},\ \forall x,b,y \\ \nonumber
&\quad \widetilde M_{b|y}\geq 0,\ \forall b,y\\ \nonumber
&\quad \sum_b \widetilde M_{b|y} = \eta\II, \ \forall y  \\ \nonumber
&\quad N_\textbf{a} \geq 0,\ \forall \textbf{a}\in\{1,\ldots,o\}^{\times|\cM|}\\ \nonumber
&\quad \sum_\textbf{a} N_\textbf{a} = \II,
\end{align}
where $\{\widetilde M_{b|y}/\eta\}_{b,y}$ is an arbitrary measurement that provides the quantum noise $\{\tr(\rho_x\widetilde{M}_{b|y})/\eta\}_{b,x,y}$. 
Running SDP (\ref{sdp:cc}), if the optimal value obtained is $\eta^*>0$, then some amount of noise is needed and we know that $\{P(b|\rho_x,y)\}$ requires quantum communication to be distributedly sampled.

Let us now fix an arbitrary quantum noise $\{q(b|\rho_x,y) = \tr(\rho_x\widetilde{M}_{b|y})\}_{b,x,y}$ (see the footnote in App. A).
The obtained SDP admits the Lagrangian
\begin{subequations}
\begin{align}
 \mathcal{L}(\eta,\{N_{a}\}) =& \eta\left(1+\sum_{b,x,y}\mu_{bxy}\left[P(b|\rho_x,y)-q(b|\rho_x,y)\right]\right) \\ &+ \sum_{a}\tr\left(N_a\left[V_{a}+B -\sum_{b,x,y}\mu_{bxy}\rho_x\delta_{a_y,b}\right]\right)\\ &\tr(B) - \sum_{b,x,y}\mu_{bxy}q(b|x,y),
 \end{align}
 \end{subequations}
for arbitrary scalars $\lambda_{bxy}$ and operators $V_{a}\geq0, B$.
Minimising over these parameters, considering feasibility constraints and eliminating the slack variables $V_{a}$, we obtain the dual formulation
\begin{subequations}\label{sdp:qcwit}
\begin{align}
\min_{\{\mu_{bxy}\},B}&\quad \tr(B) - \sum_{b,x,y}\mu_{bxy}\tr(\rho_xM_{b|y})\\ \label{witness_condition}
\text{s.t.}&\quad \sum_{b,x,y}\mu_{bxy}\rho_x\delta_{a_y,b} \geq B,\ \forall \textbf{a}\in[o]^{\times m}\\
& \quad \sum_{bxy}\mu_{bxy}[q(b|\rho_x,y) - P(b|\rho_x,y)] = 1.
\end{align}
\end{subequations}

We now show that these two problems indeed display strong duality.
This implies that every behaviour whose distributed sampling cannot be implement only with classical communication violates a quantum communication witness, and proves Thm. 1.

\begin{proof}
Since behaviours that are not quantum realisable cannot be CC-realisable, we can restrict our proof to quantum realisable behaviours.

Let $\{P(b|\rho_x,y) = \tr(\rho_xM_{b|y})\}_{b,x,y}$ be a quantum realisable behaviour.
Let $D_t: A \mapsto (1-\eta)A+\eta\tr(A)\II/d$ be the depolarising map, with robustness $\eta\in[0,1]$.
Given the above behaviour, we obtain a strictly feasible solution for SDP (\ref{sdp:cc}) by setting $q(b|x,y)=\tr(M_{b|y})/d$ as white noise and $1>\eta>0$ close enough to 1 such that the depolarised measurements $\cM_\eta=\{D_\eta(M_{b|y})\}$ are jointly measurable.
From the connection between joint measurability and Einstein-Podolsky-Rosen steering \cite{quintino2014, uola2014}, we know this happens for $\eta$ strictly less than 1 \cite{quintino2015}.
If the corresponding mother measurement $\bN=\{N_a\}_a$ possess a zero eigenvalue, then the set $\cM_{t'}$, with $1>\eta'>\eta$, admits a mother $\bN'= \{\Phi_{\eta'/\eta}(N_a)\}$, with strictly positive elements.
Hence, by Slater's condition \cite{boyd2004} the SDPs (\ref{sdp:cc}) and (\ref{sdp:qcwit}) present strong duality.

Therefore, if the behaviour can be distributedly sampled with classical communication, the primal SDP yields $\eta^*\leq0$.
Strong duality ensures that the optimal solution for the dual matches $\eta^*$, being also $\leq0$, and thus we have
\begin{equation}\label{witness}
\sum_{b,x,y}\mu_{bxy}\tr(\rho_xM_{b|y}) \geq \tr(B),
\end{equation}
for some real coefficients $\{\mu_{bxy}\}_{b,x,y}$ and some matrix $B$ acting on $\cH$ satisfying Eq. (\ref{witness_condition}).
Moreover, $R_{NCC} = \eta^* = \max_{W_{QC}}\beta^*-\sum_{b,x,y}\mu_{bxy}^*\tr(\rho_xM_{b|y})$, where this is the maximum violation of some witness for the given behaviour.
\end{proof}

As a concrete example, consider in dimension $d=2$ the distributed sampling of the statistics generated by the set of states $\hat\cS = \{\ketbra++,\ketbra yy, \ketbra00, \II/2\}$, formed by one eigenstate of each Pauli matrix $X, Y$ and $Z$ and the maximally mixed state, and the set of measurements $\hat\cM=\{\bM_x,\bM_y\}$ associated to $X$ and $Y$.
Running SDP (\ref{sdp:qcwit}) with $q$ set as white noise, we obtain the quantum communication witness $W_{QC}$ given by
\begin{align*}
&\frac{-1}{2}\left[P(1|\rho_1,1) + P(1|\rho_3, 2)\right] \\ &+ \frac12\left[P(2|\rho_1,1) + P(2| \rho_3, 2)\right]\\ 
+ &\frac{1}{2\sqrt2}\left[P(1|\rho_2,1) + P(2|\rho_2,1)\right]\\& + P(1|\rho_4,2) - P(2|\rho_4,1) \geq \frac{-1}{2\sqrt2}.
\end{align*}
The behaviour yielded by $\hat\cS$ and $\hat\cM$ violates it up to $1/(2\sqrt2)-1$.

Notice that $\cS$ spans the set of Hermitian operators acting in $\CC^2$, therefore every incompatible set of measurements is detected by some $W_{QC}$ constructed from $\cS$.

\section{C. Proofs of Theorems 2 and 3}

Here we restate and prove theorems from the main text. 

\setcounter{theorem}{1}
\begin{theorem}
A set of measurements $\cM\subset\MM(\cH)$ is compatible if and only if the behaviour $\{\tr (\rho_x M_{b|y});\ \rho_x\in\cS,\ \bM_y\in\cM\}$ admits a distributed sampling realisation with classical communication, for any set of states $\cS\subset\mathbb{S}(\cH)$ that spans $\LL(\cH)$.
Morevover, for any distributedly sampled quantum behaviour $P=\{P(b|\rho_x, y)=\tr(\widetilde{\Lambda}(\rho_x)\widetilde{M}_{b|y})\}\in\textbf{Q}$, we have 
\begin{equation}\label{bounds}
R_{NCC}(P) \leq R_{I}(\widetilde{\cM})\ \text{and}\ R_{NCC}(P) \leq R_{NSB}(\widetilde{\Lambda}),
\end{equation} 
where $\widetilde{\cM}$ and $\widetilde{\Lambda}$ are the uncharacterised measurements and communication channel, respectively, used in the sampling of $P$.
The equality holds in the first case if $\cS$ spans $\LL(\cH)$ and in the second case if, besides that, the measurements $\{\widetilde{\Lambda}^\dagger(\widetilde{M}_{b|y})\}_{b,y}$ present the greatest generalised robustness of incompatibility in its dimension.
\end{theorem}
\begin{proof}
Let $\cS\in\cS(\cH)$ be an arbitrary set of states and $\rho_x\in\cS$.
Assuming $\cM$ to be jointly measurable, Alice can perform the mother measurement $\bN$ on $\rho_x$ and send the obtained outcome $a$ to Bob, who applies $f(\cdot|y,a)$, one of the post-processing distributions that accompanies $\bN$.
The statistics generated are described by 
\begin{equation}
P(b|\rho_x,y)=\sum_{a=1}^{n} \tr (\rho_x N_a) f(b|y,a),
\end{equation}
for any $b\in[o],\rho_x\in\cS,y\in[m]$, where the equality is guaranteed, independently of $\rho_x$, by the joint measurability hypothesis
\begin{equation}\label{eq:jm}
M_{b|y} = \sum_{a=1}^{n} N_a f(b|y,a),
\end{equation}
for all $y\in[m]$ and $b\in[o]$.
It follows that $R_{NCC}(P) \leq R_{I}(\cM)$.

On the other hand, suppose that exists $\bN$ and $\{p(\cdot|y,a)\}_{y,a}$ such that $\sum_a \tr(N_a\rho)p(b|y,a)=\tr(\rho_x M_{b|y})$ for all $(\rho_x,\bM_y)\in\cS\times\cM$.
If $\cS$ spans $\LL(\cH)$, this implies the relation between measurements operators given in Eq. (\ref{eq:jm}).
Hence $\cM$ is jointly measurable, admitting $\bN$ as a mother measurement and $\{f(\cdot|y,a)\}_{y,a}$ as post-processing maps.
Thus, $R_{NCC}(P) \leq R_{I}(\cM)$ can be saturated for informationally complete sets $\cS$.

Finally, we can write $P(b|\rho_x,y) = \tr(\widetilde{\Lambda}(\rho_x)\widetilde{M}_{b|y}) = \tr(\rho_x\widetilde{\Lambda}^\dagger(\widetilde{M}_{b|y}))$ to describe the effective communication channel $\widetilde{\Lambda}$ (which comprehends Alice's preparation on $\rho_x$ and is a quantum-classical channel in the classical communication case) and the effective measurements $\{\widetilde{M}_{b|y}\}$ that generate the behaviour.
For large enough $\eta$, there exists some noise-channel $\Gamma$ such that $(1-\eta)\Lambda+\eta\Gamma$ is steering-breaking.
At his point, $\{[(1-\eta)\Lambda+\eta\Gamma]^\dagger(\widetilde{M}_{b|y})\}$ is compatible and its corresponding statistics are CC-realisable, by the first part of the theorem.
Therefore, $R_{NSB}(\widetilde{\Lambda})\geq R_{NCC}(P)$, which is saturated if $\{\rho_x\}_x$ spans $\LL(\cH)$ -- thus detecting the standard incompatibility of the underlying measurements $\widetilde{\Lambda}(\widetilde{\cM})$ -- and if those are the most incompatible acting in $\cH$ -- thus implying that the noisy channel would break the incompatibility of any other set of measurements as well.
\end{proof}

\begin{theorem}
For any measurement incompatibility witness $W_{MI}=(\{F_{by}\}, \gamma)$ there exists a set of states $\{\rho_x\}_{x}$ and a quantum communication witness $W_{QC}=(\{\mu_{bxy}\},\beta)$ in the corresponding distributed sampling scenario that detects the incompatibility of the same sets of measurements as $W_{MI}$.
\end{theorem}

\begin{proof}
Let $\{\rho_x\}_{x}$ be a set of states that spans $\{F_{by}\}$, i.e. satisfies
\beq
F_{by} = \sum_x \lambda^{by}_x\rho_x
\eeq
for some real coefficients $\{\lambda_{x}^{by}\}_{b,x,y}$.
Such a set always exists, since any Hermitian operator can be written as the difference between two positive operators, which can be renormalised to be trace-one.
Applying this procedure to a basis of the space of Hermitian operators acting in the underlying Hilbert space provides a set of quantum states that span the set of Hermitian operators.

Consider now the expression
\begin{equation}\label{bla}
\sum_{b,x,y}\lambda_x^{by}\tr(\rho_xM_{b|y}).
\end{equation}
If the behaviour $\{\tr(\rho_xM_{b|y})\}_{b,x,y}$ admits a CC-sampling, there exists a measurement $\bN$ and response functions $f(\cdot|y,a)$ such that we can rewrite (\ref{bla}) as
\begin{subequations}
\begin{align}
\sum_{b,x,y}\lambda_x^{by}\tr(\rho_xM_{b|y}) =& \sum_{b,x,y}\lambda_x^{by}\sum_a\tr(\rho_xN_a)f(b|y,a) \\
=& \sum_{b,y}\tr\left(\sum_x \lambda_x^{by}\rho_x\left[\sum_aN_af(b|y,a)\right]\right)\\
=& \sum_{b,y}\tr(F_{by}\left[\widetilde M_{b|y}\right]),
\end{align}
\end{subequations}
for some set of measurements $\{\widetilde M_{b|y} := \sum_a N_af(b|y,a)\}_{b,y}$, which is compatible by definition.
Since $(\{F_{by}\}, \gamma)$ is a incompatibility witness, the expressions above are lower-bounded by $\gamma$.

By definition, $(\{F_{by}\}, \gamma)$ detects some incompatible set of measurements $\{\hat M_{b|y}\}_{b,y}$, hence $(\{\lambda_{bxy}\}, \gamma)$ detects the quantum communication in the distributed sampling of the behaviour $\{\tr(\rho_x\hat M_{b|y})\}$.
Therefore, $(\{\lambda_{b}^xy\}, \gamma)$ defines a quantum communication witness, which detects the same incompatible measurements (with the trusted states $\rho_x$) as $(\{F_{by}\}, \gamma)$ by construction.
\end{proof}

\section{D. Relation with state discrimination}\label{app:witcomparison}

We start by pointing that if a set of states $\cS=\{\rho_x\}_x$ can be perfectly discriminated, then Alice can identify the label $x$ and send it Bob, who will again hold both inputs.
Therefore, perfect discrimination implies a CC-realisation in our distributed sampling task.
Consequently, a violation of a quantum communication witness detects not only the quantum communication needed for realising the behaviour and the incompatibility of the measurements, but also that $\cS$ is not perfectly discriminable.

It was shown in Thm. 3 that every incompatibility witness corresponds to a quantum communication witness.
As proved in Refs. \cite{carmeli2018, uola2018, skrzypczyk2019}, the former can also be phrased in terms of a discrimination game of suitable ensembles.
Hence, we explicit here how to translate from one formulation to the other.
Notice that the distributed sampling related to some set of states $\{\rho_x\}$ certifies as much incompatibility as the discrimination of a \textit{different} set of states $\{\sigma^y_b\}$, reflecting the fact that the translation is non-trivial and the tasks are intrinsically different.

Recall that our witnesses are given by 
\begin{equation}\label{wit2}
\sum_{b,x,y}\mu_{bxy}\tr(\rho_xM_{b|y}) \geq \tr(B),
\end{equation}
where the sum runs over $b=1,\ldots,o$, $x=1,\ldots,|\cS|$ and $y=1,\ldots,m$.
Applying the reasoning depicted in Thm. 1 of \cite{carmeli2018}, we find that a violation of (\ref{wit2}) provides the advantage the set $\cM=\{M_{b|y}\}_{b,y}$ presents over compatible sets when discriminating states from the subensembles $\mathcal{E}	_1=(\{\sigma^1_{b}\},p^1(b))_b,\ldots,\mathcal{E}_{m}=(\{\sigma^{m}_{b}\},p^{m}(b))_b$ given by
\begin{equation}
p^y(b)\sigma^y_b = \alpha(\sum_x\lambda_{bxy}\rho_x + \nu\II),
\end{equation}
where $\nu = \sum_{b,y}\max_x |\lambda_{bxy}|$ and $\alpha = \left(\sum_{bxy}\lambda_{bxy}+o\cdot m\cdot\nu /d\right)^{-1}$. 

Coming back to the example of the witness $W_{QC}$ in Appendix B, measurements that provide its maximum violation can optimially discriminate the ensembles given by
\begin{align*}
&\sigma^1_1 = \frac{\II-0.0551(X+Y)}{2},& &p^1(1) = 0.2711\\
&\sigma^1_2 = \frac{\II+0.0551(X-Y)}{2},& &p^1(2) = 0.2289\\
&\sigma^2_1 = \frac{\II}{2},& &p^2(1) = 0.2711\\
&\sigma^2_2 = \frac{\II+0.1306\cdot Y}{2},& &p^2(2) = 0.2289. 
\end{align*}

\section{E. Depolarising channels and restricted measurement compatibility}

In the case where Alice and Bob share a perfect quantum channel $\Lambda: \rho \mapsto \rho$, distributed sampling is a trivial task, since Alice can send Bob her input-state $\rho_x$.
In Fig. 3, we considered the case where both players exchange imperfect quantum communication, represented by a depolarising quantum channel given by $D_t: \LL(\CC^2)\rightarrow \LL(\CC^2),\ \rho \mapsto t\rho+(1-t)\II/2$, for a fixed transmittance rate $t=1-\eta\in[0,1]$.
Thus, the parameter $\eta$ marks the amount of white noise acquired in the communication.
Also, the depolarising channel is self-adjoint, meaning that $\tr[D_t(\rho)M_b] = \tr[\rho D_t(M_b)]$, for any state $\rho$ and measurement element $M_b$.

We now calculate the critical parameter $t$ that makes $D_t$ replaceable by a classical channel, in the distributed sampling context.
That will be the threshold for our quantum communication detection, and the white-noise robustness of non-steering-breaking of the identity channel.
Our main tool to answer this question is a variation of SDP (\ref{sdp:cc}) in which we fix the noise to be completely random (white noise), $\widetilde M_{b_y} = \tr(M_{b|y})\II/d$, where $\{M_{b|y}\}$ is the underlying measurement that yielded the behaviour. 
At the obtained critical transmittance rate $t^*=1-\eta^*$, the transmitted information is classical enough to camouflage any quantumness from the channel.
The quantity $\eta^*$ is then defined to be the behaviour's \textit{white-noise robustness of non-CC-realisability}, denoted by $R_{NCC}^wn(P)$.

Let $\cH= \mathbb{C}^2$ and $\cM$ be the set of qubit measurements associated to the Pauli observables $\sigma_x,\sigma_y,\sigma_z$.
For any $\cS=\{\rho_x\}_x$, we denote $D_t(\cS) \equiv \{D_t(\rho_x)\}_x$, and similarly for $D_t(\cM)$.

If $\cS_1=\{\ketbra00, \ketbra++\}$, where
\begin{equation}
\ket + = \frac{\ket0 + \ket1}{\sqrt{2}},
\end{equation}
then by SDP (\ref{sdp:cc}) we see that the behaviour generated by depolarised states $(D_{t_1^*}(\cS_1), \cM)$ is CC-realisable at the critical parameter $t_1^* = 0.9449$.
Similarly, for $\cS_2=\{\ketbra00, \ketbra++, \ketbra rr\}$, where 
\begin{equation}
\ket r = \frac{\ket0 + i\ket1}{\sqrt{2}}
\end{equation}
is an eigenstate of $\sigma_y$, we obtain $t_2^* = 0.8165 \approx \sqrt{2/3}$.
For $\cS_3 = \{\ketbra00, \ketbra++, \ketbra rr, \II/2\}$ we achieve the critical parameter $t_3^* = 0.5774 \approx 1/\sqrt3$. 
We note that the set $\cS_3 $ spans the qubit Hermitian operators space, and the critical value $t_3^* = 1/\sqrt3$  can also be obtained analytically \cite{heinosaari2008}.

From the point of view of measurement incompatibility on a restricted set of states, defined in the main text, we can interpret these $t^*$ as the critical parameters for which the set $\cM$ becomes compatible on each of the sets $\cS_1, \cS_2$, and $\cS_3$.
Since $\cS_3$ spans $\LL(\CC^2)$, we have that $D_{t^*_3}(\cM)$ is compatible in the standard sense.
By Thm. 2, this in turn implies that the behaviour generated by $(D_{t_3^*}(\cS),\cM)$ is CC-realisable for any set of states $\cS$, hence $D_{t^*_3}$ does not offer any advantage over classical channels.

The general qudit depolarising channel $D_t$ is known to be incompatibility-breaking for qudit projective measurements if and only if $t\leq t^{proj}_{d}:=\frac{1}{d-1} \left(-1+\sum_{k=1}^d\frac{1}{k}\right)$ \cite{wiseman2007,quintino2014,uola2014,almeida2007,heinosaari2015}.
Thus the qubit critical transmittance is $1/2$. 
This implies that if $t\leq t^{proj}_{d}$, the set of states $\rho_x$ received by Alice lies in $\LL(\mathbb{C}^d)$, and Bob performs projective measurements, then no quantum communication can be certified. 
On the other hand, if $t>t^{proj}_{d}$ and the set of states $\rho_x$ received by Alice spans $\LL(\mathbb{C}^d)$, then any incompatible set of projective measurements in $\LL(\mathbb{C}^d)$ implemented by Bob certifies quantum communication.

For general POVMs, by noticing that the local hidden variable model of Ref. \cite{almeida2007} can be transformed into a local hidden state model (``steering model'') and using the connection between joint measurability established in Ref. \cite{quintino2014,uola2014}, we can show that the qudit depolarising channel $D_t$ is incompatibility-breaking for all measurements if  $t\leq t^{all}_d:=\frac{(3d-1)(d-1)^{d-1}}  {(d-1)d^d}$. 
Hence when $t\leq t^{all}_d$ and the set of states $\{\rho_x\}_x$ lies in $\LL(\mathbb{C^d})$, all behaviours admit a CC-realisation, regardless of the measurements performed by Bob.

For the case where the number of measurements performed by Bob is a finite $n\in \mathbb{N}$, it follows from Theorem 2 that a channel can be used to certify QC if and only if it is $n$-incompatibility breaking \cite{heinosaari2015}, \textit{i.e.} it breaks the incompatibility of any $n$ measurements. 
For this situation, there exists a systematic numerical method with a sequence of algorithms that converge to the exact critical value for the depolarising channel to be $n$-incompatibility breaking \cite{bavaresco2017}. 
This numerical approach provides upper and lower bounds for the $n$-incompatibility-breaking critical value of $D_t$ in finitely many steps.

We also remark that $D_t$ is entanglement-breaking if and only if $t\leq \frac{1}{d+1}$ \cite{horodecki2003}.
Also, a channel is entanglement-breaking if and only if it admits a measure-and-prepare realisation, that is, it can be described as $\rho \mapsto \sum_b\tr(\rho N_{b})\sigma_b$, for some measurement $\bN$ and states $\{\sigma_b\}_b$. 
This description provides a clear recipe for a CC-realisation of distributed sampling that simulates these channels.

All these calculations can be found in the repository \url{www.github.com/guerinileonardo/DS_QCwit_MI}.

\end{document}